\documentclass[10pt,conference]{IEEEtran}

\usepackage{theorem}
\usepackage{amsmath}
\usepackage{graphicx} 
\usepackage{stmaryrd}
\usepackage{bbm}

\IEEEoverridecommandlockouts

\newcommand{\dmin}{\ensuremath{d}}
\newcommand{\erasure}{\ensuremath{\vartimes}}
\newcommand{\Erf}{\ensuremath{\mathrm{Erf}}}
\newcommand{\Erfc}{\ensuremath{\mathrm{Erfc}}}
\renewcommand{\vec}[1]{\ensuremath{\mathbf{#1}}}
\newcommand{\R}{\ensuremath{\mathbbm{R}}}

\newcommand{\up}{\ensuremath{\underline{p}}}
\newcommand{\op}{\ensuremath{\overline{p}}}
\newcommand{\ul}{\ensuremath{\underline{l}}}
\newcommand{\ol}{\ensuremath{\overline{l}}}
\newcommand{\ut}{\ensuremath{\underline{t}}}
\newcommand{\ot}{\ensuremath{\overline{t}}}

{\theoremstyle{plain}
   }

{\theoremstyle{break}
   }

{\theoremstyle{plain}
   \newtheorem{theorem}{Theorem}}

{\theoremstyle{plain}
   }

{\theoremstyle{plain}
   }

{\theoremstyle{plain}
   }

\begin{document}

\title{On Generalized Minimum Distance Decoding Thresholds for the AWGN Channel}

\author{
\authorblockN{Christian Senger, Vladimir R. Sidorenko}\thanks{This work has been supported by DFG (German Research Council) under grants BO~867/15, and BO~867/17. Vladimir R. Sidorenko is on leave from IITP, Russian Academy of Sciences, Moscow, Russia.}
\authorblockA{\small Inst. of Telecommunications and Applied Information Theory\\
Ulm University, Ulm, Germany \\
\{christian.senger$\;\vert\;$vladimir.sidorenko\}@uni-ulm.de}
\and
\authorblockN{Victor V. Zyablov}
\authorblockA{\small Inst. for Information Transmission Problems\\
Russian Academy of Sciences, Moscow, Russia \\
zyablov@iitp.ru}
}

\maketitle

\begin{abstract}
We consider the Additive White Gaussian Noise channel with Binary Phase Shift Keying modulation. Our aim is to enable an algebraic hard decision Bounded Minimum Distance decoder for a binary block code to exploit soft information obtained from the demodulator. This idea goes back to Forney \cite{forney:1966a}, \cite{forney:1966b} and is based on treating received symbols with low reliability as erasures. This erasing at the decoder is done using a threshold, each received symbol with reliability falling below the threshold is erased. Depending on the target overall complexity of the decoder this pseudo--soft decision decoding can be extended from one threshold $T$ to $z>1$ thresholds $T_1<\cdots<T_z$ for erasing received symbols with lowest reliability. The resulting technique is widely known as Generalized Minimum Distance decoding. In this paper we provide a means for explicit determination of the optimal threshold locations in terms of minimal decoding error probability. We do this for the one and the general $z>1$ thresholds case, starting with a geometric interpretation of the optimal threshold location problem and using an approach from \cite{zyablov:1970}.
\end{abstract}

\section{Introduction}\label{section:introduction}
The concept of concatenated codes was introduced by Forney in 1966 \cite{forney:1966a}. Concatenated codes consist of an inner and an outer code, a decoder for the concatenated code includes their associated decoders. Encoding is done such that the information block to be transmitted is first encoded using the outer code and then the symbols of the resulting outer codeword are encoded using the inner code. At the receiver side first the decoder for the inner code calculates estimates for the outer codeword symbols. Then, the decoder for the outer code tries to reconstruct the transmitted codeword utilizing the estimates from the inner decoder as inputs. In his original work, Forney proposed {\em Generalized Minimum Distance (GMD)} decoding, which extends simple single--trial decoding of concatenated codes to multiple decoding trials. More precisely, Forney specified GMD decoding for an integer $z>(\dmin-1)/2$ of decoding trials, where $\dmin$ is the minimum Hamming distance of the outer code. For smaller values of $z$, Weber and Abdel--Ghaffar later introduced the term {\em reduced} GMD decoding \cite{weber_abdel-ghaffar:2003}. GMD decoding relies on an outer error/erasure decoder and works as follows. In each decoding trial, an increasing set of most unreliable symbols obtained from the inner decoder are erased. The resulting intermediate word is fed into the outer error/erasure decoder, which calculates an outer codeword estimate. Potentially, each decoding trial results in a different outer codeword estimate so some means of selecting the ''best'' estimate needs to be provided.

Let the number of performed decoding trials be $z$. We do not distinguish between reduced GMD decoding and full GMD decoding and allow $z$ to be any non-zero natural number independent of the code parameters. In practice, erasing of the most unreliable symbols is accomplished using a set of real--valued thresholds $\{T_1, \ldots, T_z\}$ with $T_1<\cdots< T_z$. If the reliability value of a symbol falls below threshold $T_i$ in the \mbox{$i$-th} decoding trial, then this symbol is marked as erasure in this trial. The threshold version of GMD decoding was presented by Blokh and Zyablov \cite{blokh_zyablov:1982}.

In this paper we consider a special case of a code concatenation, i.e. the case where the inner ''code'' is {\em Binary Phase Shift Keying (BPSK)} modulation and the outer code is a linear binary code with an error/erasure {\em Bounded Minimum Distance (BMD)} decoder. Such decoders are well--known for certain important code classes, e.g. for {\em Bose--Chaudhuri--Hocquenghem (BCH)} codes \cite{blahut:2003}.

Our work is organized as follows. In Section \ref{section:definitions} we give basic definitions and notations that are used in the remainder of the paper. Section \ref{section:1T} considers the most simple case of (reduced) GMD decoding, i.e. error/erasure BMD decoding with one single threshold. Its optimal location is derived using a geometric approach. Note that we use ''optimal'' as an abbreviation for ''optimal in terms of minimal decoding error probability''. In Section \ref{section:zT} we consider the general case of $z>1$ thresholds before we finally wrap up the paper with conclusions and further research perspectives in Section \ref{section:conclusions}.

\section{Definitions and Notations}\label{section:definitions}

Assume an {\em Additive White Gaussian Noise (AWGN)} channel with BPSK modulation, let the transmitted symbols be w.l.o.g. $x\in\{-1, +1\}$, i.e. the modulator performs for every transmitted binary value $c\in\{0, 1\}$ the operation $x=(-1)^c$ and the transmit signal power is fixed to $E_s=1$. Hence, the standard deviation of the AWGN channel is $\sigma=\sqrt{N_0/2}$. We define the probability that for given $\sigma$ a transmitted symbol $x$ results in a received symbol $y$ within the real interval $[a, b]$ as
\begin{equation*}
p_\sigma(a, b):=\int_a^b \frac{1}{\sqrt{2\pi}\sigma} \,\mathrm{exp}\left(-\frac{(\chi-x)^2}{2\sigma^2}\right)\;d\chi.
\end{equation*}
For simplicity we also define the negative logarithmic probability
\begin{equation*}
l_\sigma(a, b):=-\mathrm{ln}\left(p_\sigma(a, b)\right).
\end{equation*}

As outer code we assume a linear binary $(n, k, \dmin)$ code $\mathcal{C}$ with code length $n$, dimension $k$ and minimum Hamming distance $\dmin$. An error/erasure BMD decoder for $\mathcal{C}$ can decode error patterns with $\tau$ erasures and $\epsilon$ errors as long as
\begin{equation}\label{eqn:cap}
2\epsilon+\tau<\dmin.
\end{equation}

A codeword $\vec{c}=(c_0, \ldots, c_{n-1})\in\mathcal{C}$ is mapped to a vector $\vec{x}=(x_0, \ldots, x_{n-1})\in\{-1, +1\}^n$ by the modulation function described above. At the receiver side, the vector $\vec{y}=(y_0, \ldots, y_{n-1})\in\R^n$ is received. For each received symbol holds $y_j=x_j+\xi$, where $\xi$ is the realization of a Gaussian noise process with mean $x_j$ and standard deviation $\sigma$.

\section{The Single Threshold Case}\label{section:1T}

We start our considerations with the case of one single threshold $0\leq T\leq 1=E_s$. This means that the following quantization--and--erasing function is applied to any received symbol $y_j$.
\begin{equation*}
  \phi_{T}:=\left\{\begin{array}{rcl}
    \R & \longrightarrow & \{0, 1\}\cup\erasure\\
    y_j & \longmapsto & \left\{\begin{array}{cl}
      1 & ;\;\text{if}\;y_j<-T\\
      0 & ;\;\text{if}\;y_j>T\\
      \erasure & ;\;\text{if}\;-T\leq y_j \leq T
    \end{array}\right.\\
  \end{array}\right..
\end{equation*}
The obvious extension of $\phi_T$ to vectors is
\begin{equation*}
\phi_T(\vec{y}):=\left(\phi_T(y_0), \ldots, \phi_T(y_{n-1})\right).
\end{equation*}

Note that since $\mathcal{C}$ is a linear code and the threshold location is symmetric, we can restrict our considerations in the following w.l.o.g. to the case $\forall\,j=0, \ldots, n-1: x_j=+1$, i.e. transmission of the all--zero codeword.

Consider the probability $P_\sigma$ that the decoder produces an error, i.e. the probability that it either returns no codeword or a wrong codeword. We make use of the abbreviated notation $p_x:=p_\sigma(-T, T)$ and $p_e:=p_\sigma(-\infty, -T)$ for the erasure and error probability, respectively. Similarly, we define the negative logarithms $l_x:=-\ln(p_x)$ and $l_e:=-\ln(p_e)$.
\begin{multline}\label{eqn:exact1T}
P_\sigma= \sum_{\tau=0}^{n}  \sum_{\epsilon=t_\tau}^{n-\tau} \binom{n}{\tau, \epsilon, n-\tau-\epsilon}\cdot\\
\cdot p_x^\tau p_e^\epsilon \left(1-p_x-p_e\right)^{n-\tau-\epsilon},
\end{multline}
where $t_\tau:=\left\lceil\frac{\dmin-\tau}{2}\right\rceil$. For good channel conditions, i.e. small values of $\sigma$, we obtain the approximation
\begin{equation*}
P_\sigma\approx \max_{0\leq\tau\leq\dmin}\left\{
\binom{n}{\tau, t_\tau, n-\tau-\epsilon}\, p_x^\tau p_e^{t_\tau} \right\}.
\end{equation*}
Note that the last term in (\ref{eqn:exact1T}) can be neglected since it is close to one. Transforming this approximation into negative logarithmic form we obtain
\begin{multline*}
-\ln(P_\sigma)\approx \min_{0\leq\tau\leq\dmin}\left\{
\tau\, l_x+ t_\tau\, l_e -\right.\\
\left. -\ln(2)\left(n\, H(\tau/n)+ (n-\tau)\, H\left(\frac{t_\tau}{n-\tau}\right)\right)\right\},
\end{multline*}
where $H(\cdot)$ denotes the binary entropy function. Since it only assumes values between 0 and 1 and $l_x$ and $l_e$ tend to infinity for small $\sigma$, we can further approximate
\begin{equation}\label{eqn:approx}
-\ln(P_\sigma)\approx \min_{0\leq\tau\leq\dmin}\left\{ \tau\, l_x+ t_\tau\, l_e \right\}.
\end{equation}

Now we return to the non-abbreviated notation and define the goal function
\begin{equation}\label{eqn:1Tgoalfunction}
  g_\sigma(\tau, T):=\tau\, l_\sigma(-T, T)+\frac{\dmin-\tau}{2} \,l_\sigma(-\infty, -T).
\end{equation}
We omit the ceiling operation from $t_\tau$ to obtain a function which is linear in $\tau$. By means of (\ref{eqn:approx}) we observe that the minimum of the goal function over $\tau$ approximates the negative logarithmic decoding error probability as long as the channel standard deviation $\sigma$ is small. The behavior of the goal function for several thresholds is depicted in Figure~\ref{fig:1Tgoalfunctions}. The number of erasures $\tau$ is spread on the abscissa and each straight line represents one threshold $0\leq T\leq 1=E_s$, the minimum of each straight line represents the approximated negative logarithmic error probability for this specific threshold. The decoder's aim is to select the threshold $T$ such that the minimum is maximized since this yields the minimal decoding error probability.

The following theorem provides a necessary and sufficient criterion for the optimal high--SNR threshold $T_\sigma$.

\begin{theorem}\label{thm:Optimal1T}
For good channel conditions, i.e. small channel standard deviation $\sigma$, $T_\sigma$ is the optimal threshold if and only if the following equation is fulfilled.
\begin{equation}\label{eqn:solve1T}
  \sqrt{p_\sigma(-\infty, -T_\sigma)}  =  p_\sigma(-T_\sigma, T_\sigma).
\end{equation}
\end{theorem}
\begin{proof}
Since the goal function is linear in $\tau$, it assumes its minimum at one of the two extremal points $g_\sigma(0, T)$ and $g_\sigma(\dmin, T)$ which means that (\ref{eqn:approx}) reduces to
\begin{equation*}
-\ln(P_\sigma)\approx \min \left\{g_\sigma(0, T), g_\sigma(\dmin, T)  \right\}.
\end{equation*}
Let $T_\sigma$ be such that $g_\sigma(0, T_\sigma)=g_\sigma(\dmin, T_\sigma)$. Inserting the definition of the goal function shows that this is equivalent to
\begin{equation}\label{eqn:equality}
p_\sigma(-\infty, -T_\sigma)^{\frac{\dmin}{2}}=p_\sigma(-T_\sigma, T_\sigma)^{\dmin}.
\end{equation}
Assume that threshold $T'\neq T_\sigma$ is optimal. This gives
\begin{eqnarray*}
p_\sigma(-T', T') & = & p_\sigma(-T_\sigma, T_\sigma) +\Delta\;\mathrm{and}\\
p_\sigma(-\infty, -T') & = & p_\sigma(-\infty, -T_\sigma)-\Delta,
\end{eqnarray*}
where $\Delta>0$ if $T'>T_\sigma$ and $\Delta<0$ if $T'<T_\sigma$ since both
\begin{equation*}
p_\sigma(-\infty, -T_\sigma)+p_\sigma(-T_\sigma, T_\sigma)+p_\sigma(T_\sigma, \infty)=1
\end{equation*}
and
\begin{equation*}
p_\sigma(-\infty, -T')+p_\sigma(-T', T')+p_\sigma(T', \infty)=1
\end{equation*}
must be fulfilled. If we transform (\ref{eqn:approx}) back to the non--logarithmic domain we obtain
\begin{equation}\label{eqn:max}
P_\sigma\approx \max \left\{ p_\sigma(-\infty, -T)^{\frac{\dmin}{2}}, p_\sigma(-T, T)^{\dmin} \right\}.
\end{equation}
By using threshold $T'$, we increase one of the two expressions in (\ref{eqn:max}) and thereby also the maximum of both expressions. But this means that the decoding error probability is increased and thus $T'\neq T_\sigma$ cannot be the optimal threshold. Hence, $T_\sigma$ is optimal and the statement of the theorem is proved.
\end{proof}

Theorem~\ref{thm:Optimal1T} allows for the following geometric interpretation. The optimal threshold $T_\sigma$ is the specific threshold for which the goal function is a perfectly horizontal line in Figure~\ref{fig:1Tgoalfunctions}.

\begin{figure}[htbp]
\centering
\includegraphics[width=252pt]{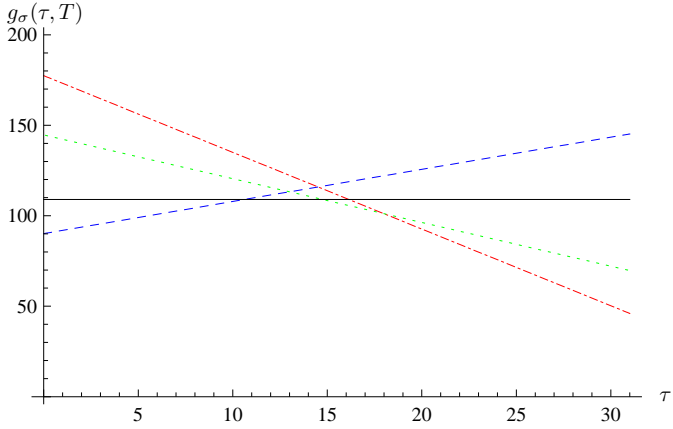}
\caption{Four exemplary instances of the goal function for $\sigma=0.4$ and $\dmin=31$. The minimum of each instance represents the negative logarithmic decoding error probability achievable with the specific threshold.}
\label{fig:1Tgoalfunctions}
\end{figure}

Figure~\ref{fig:1Tlocation} shows in the upper curve the optimal high--SNR threshold $T_\sigma$ for SNR values between $0$ and $20\,\mathrm{dB}$, the plot was obtained by numerically solving equation (\ref{eqn:solve1T}). Each point on the curve represents the optimal threshold for the specific SNR value, i.e. the threshold for which the goal function (\ref{eqn:1Tgoalfunction}) is independent of $\tau$ and thereby a perfect horizontal line in Figure~\ref{fig:1Tgoalfunctions}.

\begin{figure}[htbp]
\centering
\includegraphics[width=252pt]{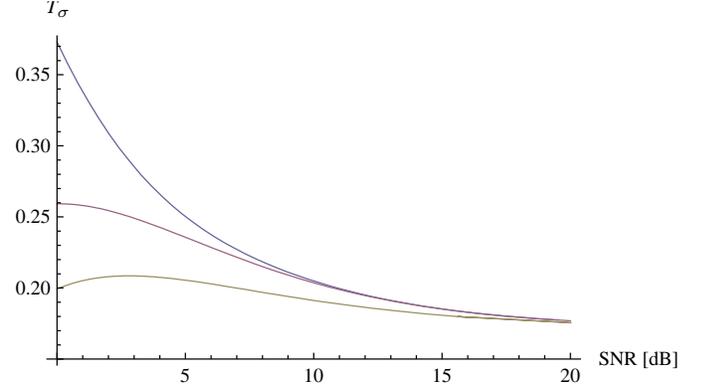}
\caption{Optimal threshold location $T_\sigma$ for SNR values between $0$ and
$20\,\mathrm{dB}$, $\sigma=\sqrt{\frac{1}{2}\,10^{-\frac{\mathrm{SNR}}{10}}}$.
The upper curve is the numerically calculated optimal high--SNR threshold given by
Theorem~\ref{thm:Optimal1T} and the middle curve is the analytic high--SNR threshold
from (\ref{eqn:1Tlocation}). The lower curve is the general optimal threshold for the full SNR range and was obtained 
by numerically minimizing (\ref{eqn:exact1T}) for a binary code with length $127$ and minimum distance $63$.}
\label{fig:1Tlocation}
\end{figure}

Obtaining an analytic solution for equation (\ref{eqn:solve1T}) is non--trivial since it essentially means solving
\begin{equation*}
\left(
\Erf\left( \frac{T-1}{\sqrt{2}\sigma} \right)
+ \Erf\left( \frac{T+1}{\sqrt{2}\sigma} \right)
\right)^2=
2\, \Erfc\left( \frac{T+1}{\sqrt{2}\sigma} \right),
\end{equation*}
where
\begin{equation*}
  \Erf(\alpha):=\frac{2}{\sqrt{\pi}}\int_0^\alpha e^{-\chi^2}\;d\chi
\end{equation*}
is the error function and $\Erfc(\alpha):=1-\Erf(\alpha)$ is its complementary counterpart. However, using the well--known approximation
\begin{equation*}
  \Erfc(\alpha)\approx \frac{\sqrt{2}}{\sqrt{\pi}\alpha} e^{-\frac{\alpha^2}{2}}
\end{equation*}
from \cite{forney:1966a} which is good for $\alpha>1$ we can at least for good channel conditions (i.e. small standard deviation $\sigma$) obtain the analytic solution
\begin{equation}\label{eqn:1Tlocation}
  T_\sigma:=3+3\sigma^2-\sqrt{9\sigma^4+\left(18-\ln\left(\frac{2\pi}{\sigma^2}\right)\right)\sigma^2+8},
\end{equation}
which approximates the optimal high--SNR threshold location for given $\sigma$. Figure~\ref{fig:1Tlocation} compares the numerical and the analytical optimal high--SNR threshold locations with the general optimal threshold. Note that the analytic approximation is only valid for high SNR values. This imposes no problem since the numerically calculated threshold given by Theorem~\ref{thm:Optimal1T} is also only valid in the high SNR regime.

We can utilize the analytic optimal threshold location to show the gain of single--threshold error/erasure BMD decoding over errors--only decoding for good channel conditions. If the optimal threshold is used, (\ref{eqn:max}) allows to approximate the decoding error probability by
\begin{equation*}
P_\sigma\approx p_\sigma(-\infty, -T_\sigma)^{\frac{\dmin}{2}}.
\end{equation*}
It is further well--known that the error probability of errors--only BMD decoding can be approximated by
\begin{equation*}
P_\mathrm{BMD}\approx p_\sigma(-\infty, 0)^{\frac{\dmin}{2}}.
\end{equation*}
Now we let $\sigma\rightarrow\infty$. From (\ref{eqn:1Tlocation}) we get $T_\sigma=3-2\sqrt{2}$. We can then solve
\begin{eqnarray*}
p_{\sigma_1}\left(-\infty, -3+2\sqrt{2}\right)^{\frac{\dmin}{2}} & = & p_{\sigma_2}\,(-\infty, 0)^{\frac{\dmin}{2}}\Leftrightarrow\\
\Erfc\left(\frac{2\left(\sqrt{2}-1\right)}{\sigma_1}\right) & = & \Erfc\left(\frac{1}{\sqrt{2}\sigma_2}\right)\Leftrightarrow\\
\sigma_1 & = & 2\sqrt{2}\left(\sqrt{2}-1\right)\sigma_2
\end{eqnarray*}
to see that the gain is $20\,\log_{10}\left( 2\sqrt{2}\left(\sqrt{2}-1\right) \right)\approx 1.4\,\mathrm{dB}$. This is in line with results obtained in the original works by Forney \cite{forney:1966a}, \cite{forney:1966b}.

\section{The General $z$ Thresholds Case}\label{section:zT}

We advance to the general case, where $z>1$ thresholds are used to determine which of the received symbols are considered as unreliable and thus are erased. The situation is depicted in Figure~\ref{fig:zTsketch}: We consider a set of $z$ thresholds $\mathcal{T}:=\{T_1, \ldots, T_z\}$ fulfilling $0\leq T_1<\cdots<T_z\leq 1=E_s$ and $z$ trials of error/erasure decoding for the received vector $\vec{y}$ are performed. The first one with decoder input $\phi_{T_1}(\vec{y})$, the second one with decoder input $\phi_{T_2}(\vec{y})$ and so on, where the quantization--and--erasing function is
\begin{equation*}
  \phi_{T_i}:=\left\{\begin{array}{rcl}
    \R & \longrightarrow & \{0, 1\}\cup\erasure\\
    y_j & \longmapsto & \left\{\begin{array}{cl}
      1 & ;\;\text{if}\;y_j<-T_i\\
      0 & ;\;\text{if}\;y_j>T_i\\
      \erasure & ;\;\text{if}\;-T_i\leq y_j \leq T_i
    \end{array}\right.\\
  \end{array}\right..
\end{equation*}
The result of this approach can obviously be a list of codewords. In our simplified setting, where the inner code is BPSK modulation, the selection of the best guess from this result list is straightforward -- it can be realized by applying the modulation operation to the binary symbols of all result list entries and choosing the one with the smallest Euclidean distance to the received vector $\vec{y}$. In the $z>1$ thresholds case we denote the event that none of the list entries is the originally transmitted codeword or that the list is empty as decoding error with probability $P_\sigma$.

\begin{figure}[htbp]
\centering
\includegraphics[width=252pt]{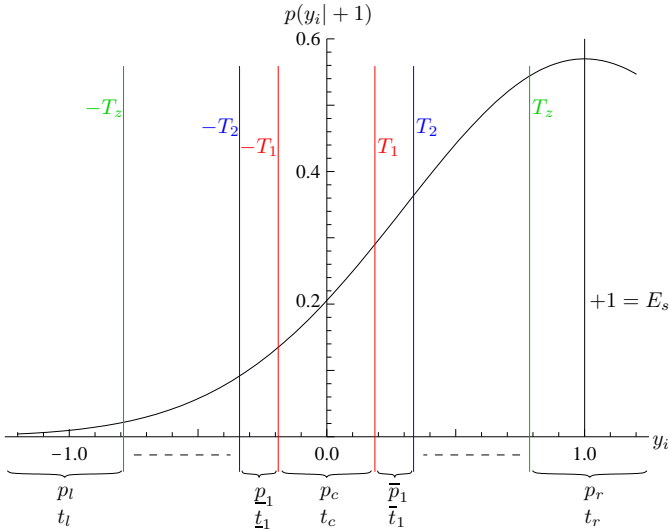}
\caption{Sketch of the threshold locations depicting the possible erasure intervals depending on thresholds $0\leq T_1<\cdots<T_z\leq 1=E_s$.}
\label{fig:zTsketch}
\end{figure}

In support of a dense notation we define the following abbreviated probabilities and their negative logarithmic counterparts.
\begin{eqnarray*}
p_l:=p_\sigma(-\infty, -T_z) & \mathrm{and} & l_l:=-\ln(p_l),\\
p_c:=p_\sigma(-T_1, T_1) & \mathrm{and} & l_c:=-\ln(p_c),\\
p_r:=p_\sigma(T_z, \infty) & \mathrm{and} & l_r:=-\ln(p_r),\\
\up_i:=p_\sigma(-T_{i+1}, -T_i) & \mathrm{and} & \ul_i:=-\ln(\up_i),\\
\op_i:=p_\sigma(T_i, T_{i+1}) & \mathrm{and} & \ol_i:=-\ln(\op_i),
\end{eqnarray*}
where $i=1, \ldots, z-1$. We also define the numbers of symbols within the received vector $\vec{y}$, that fall into the specific intervals.
\begin{eqnarray*}
t_l & := & \text{received symbols within}\;(-\infty, -T_z)],\\
t_c & := & \text{received symbols within}\;(-T_1, T_1),\\
t_r & := & \text{received symbols within}\;[T_z, \infty),\\
\ut_i & := & \text{received symbols within}\;(-T_{i+1}, -T_i],\\
\ot_i & := & \text{received symbols within}\;[T_i, T_{i+1}),
\end{eqnarray*}
where again $i=1, \ldots, z-1$. Some intervals and their corresponding abbreviated probability and number of symbols are depicted in Figure~\ref{fig:zTsketch}. With the previous definitions, the decoding error probability can be stated explicitly by
\begin{equation}\label{eqn:exactzT}
P_\sigma= \sum_{C} \binom{n}{t_l, t_c, t_r, \ut_1, \ot_1, \ldots, \ut_z, \ot_z} 
p_l^{t_l} p_c^{t_c} p_r^{t_r}
\prod_{i=1}^{z-1} \up_i^{\ot_i} \op_i^{\ot_i},
\end{equation}
where the sum is over all non-negative integers satisfying the two conditions
\begin{equation*}
C:=\left[ \begin{array}{l}
t_l+t_c+t_r+\sum_{i=1}^{z-1} (\ut_i+\ot_i)=n\quad\mathrm{and}\\
\forall\, i=1, \ldots, z:\\
\; 2(t_l+\sum_{\nu=i}^{z-1} \ut_\nu)+t_c+\sum_{\nu=1}^{i-1}(\ut_\nu+\ot_\nu)\geq \dmin
\end{array}\right].
\end{equation*}
The first condition in $C$ is obvious, it simply states that the total number of received symbols must equal the code length $n$. The second condition represents a decoding error for error/erasure BMD decoding of {\em all} input vectors $\phi_{T_i}$, $i=1, \ldots, z$. In this case, the number of errors for threshold $T_i$ is $\epsilon_{T_i}=t_l+\sum_{\nu=i}^{z-1} \ut_\nu$ and the number of erasures is $\tau_{T_i}=t_c+\sum_{\nu=1}^{i-1}(\ut_\nu+\ot_\nu)$ as can be easily seen by means of Figure~\ref{fig:zTsketch}. The second condition then follows from (\ref{eqn:cap}).

We can obtain an approximation of $P_\sigma$ if we assume that the second condition is fulfilled with equality for all thresholds $T_i\in\mathcal{T}$. For $i=1, \ldots, z-1$ we can then substract
\begin{equation*}
  2(t_l+\sum_{\nu=i+1}^{z-1} \ut_\nu)+t_c+\sum_{\nu=1}^{i}(\ut_\nu+\ot_\nu)= \dmin
\end{equation*}
from
\begin{equation*}
  2(t_l+\sum_{\nu=i}^{z-1} \ut_\nu)+t_c+\sum_{\nu=1}^{i-1}(\ut_\nu+\ot_\nu)= \dmin
\end{equation*}
and see that it holds
\begin{equation}\label{eqn:sidesequal}
\forall\,i=1, \ldots, z-1:\ut_i=\ot_i.
\end{equation}
Obeying this equality we obtain the new condition
\begin{equation*}
C^*:=\left[ 2 t_l+t_c+2 \sum_{i=1}^{z-1} \ut_i=\dmin\right]
\end{equation*}
for the sum in (\ref{eqn:exactzT}).

For good channel conditions, i.e. small values of the channel standard deviation $\sigma$, the decoding error probability can be approximated by
\begin{equation}\label{eqn:approxzT1}
P_\sigma\approx \max_{C^*}\left\{ p_l^{t_l} p_c^{t_c}
\prod_{i=1}^{z-1} (\up_i\op_i)^{\ut_i}\right\}.
\end{equation}
The term $p_r^{t_r}$ in (\ref{eqn:exactzT}) can be neglected since for small $\sigma$ it is close to one. Furthermore, (\ref{eqn:sidesequal}) is used to group the coefficients of the product under a single exponent. By transforming (\ref{eqn:approxzT1}) into negative logarithmic form we obtain
\begin{equation}\label{eqn:approxzT2}
-\ln(P_\sigma)\approx \min_{C^*}\left\{
t_l l_l+t_c l_c+\sum_{i=1}^{z-1} \ut_i(\ul_i+\ol_i)
\right\},
\end{equation}
which contains the goal function
\begin{multline}\label{eqn:zTgoalfunction}
g_\sigma(t_l, t_c, \ut_1, \ldots, \ut_{z-1}, T_1, \ldots, T_z):=\\
t_l l_l+t_c l_c+\sum_{i=1}^{z-1} \ut_i(\ul_i+\ol_i).
\end{multline}

The following theorem, whose proof exploits the linearity of the goal function in $t_l, t_c, \ut_1, \ldots, \ut_z$, provides a necessary and sufficient criterion for the optimal set of thresholds $\mathcal{T}_\sigma$.

\begin{theorem}\label{thm:OptimalzT}
For good channel conditions, i.e. small channel standard deviation $\sigma$, $\mathcal{T}_\sigma:=\{T_{1, \sigma}, \ldots, T_{z, \sigma}\}$ is the optimal set of thresholds if and only if the following system of equations is fulfilled.
\begin{eqnarray*}
  \sqrt{p_\sigma(-\infty, -T_{z, \sigma}}) & = & p_\sigma(-T_{1, \sigma}, T_{1, \sigma}),\\
  p_\sigma(-T_{1, \sigma}, T_{1, \sigma}) & = & \sqrt{p_\sigma(-T_{2, \sigma}, -T_{1, \sigma})p_\sigma(T_{1, \sigma}, T_{2, \sigma})}
\end{eqnarray*}
and
\begin{multline*}
  \forall\,i=1, \ldots, z-2:\\ p_\sigma(-T_{i+1, \sigma}, -T_{i, \sigma})p_\sigma(T_{i, \sigma}, T_{i+1, \sigma}) =\\ p_\sigma(-T_{i+2, \sigma}, -T_{i+1, \sigma})p_\sigma(T_{i+1, \sigma}, T_{i+2, \sigma}).
\end{multline*}
\end{theorem}
\begin{proof}
Due to the linearity of the goal function in $t_l, t_c, \ut_1, \ldots, \ut_z$, it assumes its minimum at one of the extremal points given by condition $C^*$, i.e. (\ref{eqn:approxzT2}) reduces to
\begin{eqnarray}\label{eqn:approxzT3}
-\ln(P_\sigma) & \approx & \min \left\{
g_\sigma\left(\frac{\dmin}{2}, 0, 0, \ldots, 0, T_1, \ldots, T_z\right), \right.\nonumber\\
 & & g_\sigma\left(0, \dmin, 0, \ldots, 0, T_1, \ldots, T_z\right),\nonumber\\
 & & g_\sigma\left(0, 0, \frac{\dmin}{2}, \ldots, 0, T_1, \ldots, T_z\right),\nonumber\\
 & & \cdots\nonumber\\
 & & \left. g_\sigma\left(0, 0, 0, \ldots, \frac{\dmin}{2}, T_1, \ldots, T_z\right)
 \right\}.
\end{eqnarray}
Let $\mathcal{T}_\sigma$ be the set of thresholds such that the value of the goal function is equal at all extremal points. Returning to the non--logarithmic representation, (\ref{eqn:approxzT3}) becomes
\begin{equation}\label{eqn:approxzT4}
P_\sigma \approx \max \left\{
p_l^{\frac{\dmin}{2}}, p_c^{\dmin}, \up_1\op_1^{\frac{\dmin}{2}}, \ldots,
\up_z\op_z^{\frac{\dmin}{2}}
\right\}.
\end{equation}
Let $\mathcal{T}'$ be a set of thresholds where at least one threshold is different than in $\mathcal{T}_\sigma$. Assume that $\mathcal{T}'$ is optimal. The only possible way for $\mathcal{T}'$ to decrease $P_\sigma$ would be to decrease all terms in (\ref{eqn:approxzT4}) simultaneously. This is impossible since the probabilities necessarily sum up to one, hence $\mathcal{T}_\sigma$ is the optimal set of thresholds and the statement is proved.
\end{proof}

\section{Conclusions and Outlook}\label{section:conclusions}

In this paper we considered a special case of (reduced) GMD decoding, i.e. transmission over an AWGN channel, BPSK modulation and error/erasure BMD decoding of a binary code. Starting from the single threshold case where only one decoding trial is performed, we generalized our considerations to the the $z>1$ thresholds case. For both cases, we derived thresholds for erasing unreliable received symbols, that are optimal in terms of the achievable minimal decoding error probability. To simplify usage of our results in practical applications we gave the approximated analytic threshold location for the single threshold case.

We showed that a gain of $1.4\,\mathrm{dB}$ over errors--only BMD decoding can be achieved with single--trial error/erasure decoding. We did not address the error probability of GMD decoding with $z>1$ thresholds in this paper. However, Forney showed that for good channel conditions the gain over errors--only decoding is approximately $3\,\mathrm{dB}$ if $z> (\dmin-1)/2$, i.e. in case of {\em full} GMD decoding.

Our work on the subject is continued with the goal to generalize the considerations from this paper to concatenated codes where the inner code is a binary block code and the outer code is a (potentially interleaved) {\em Reed--Solomon} code.

\bibliographystyle{ieeetr}

\begin{thebibliography}{1}


\bibitem{forney:1966a}
G.~D. {Forney}, {\em Concatenated Codes}.
\newblock Cambridge, MA, USA: M.I.T. Press, 1966.

\bibitem{forney:1966b}
G.~D. Forney, ``Generalized minimum distance decoding,'' {\em IEEE Trans.
  Inform. Theory}, vol.~IT-12, pp.~125--131, April 1966.

\bibitem{zyablov:1970}
V.~V. Zyablov, ``Analysis of correcting properties of iterated and concatenated
  codes,'' in {\em Transmission of Digital Information by Channels with
  Memory}, pp.~76--85, Moscow: Nauka, 1970.
\newblock In Russian.

\bibitem{weber_abdel-ghaffar:2003}
J.~H. Weber and K.~A.~S. Abdel-Ghaffar, ``Reduced {GMD} decoding,'' {\em IEEE
  Trans. Inform. Theory}, vol.~IT-49, pp.~1013--1027, April 2003.

\bibitem{blahut:2003}
R.~E. Blahut, {\em Algebraic Codes for Data Transmission}.
\newblock Cambridge: Cambridge University Press, first~ed., 2003.
\newblock ISBN 0-521-55374-1.

\bibitem{blokh_zyablov:1982}
E.~L. Blokh and V.~V. Zyablov, {\em Linear Concatenated Codes}.
\newblock Nauka, 1982.
\newblock In Russian.


\end{thebibliography}

\end{document}